\documentclass[times, 11pt]{article}

\usepackage{latex8}
\usepackage{times}
\usepackage{graphicx}
\usepackage{epsfig}
\usepackage[usenames,dvipsnames]{color}
\usepackage{latexsym}
\usepackage{amsmath}
\usepackage{amsthm}
\usepackage{amssymb}
\usepackage{algorithm}
\usepackage{algorithmic}
\usepackage{paralist}
\usepackage{multirow}

\newif\ifnotes
\notestrue 

\newcommand{\frameit}[2]{
    \begin{center}
    {\color{BrickRed}
    \framebox[3.3in][l]{
        \begin{minipage}{3in}
        \inred{#1}: {\sf\color{Black}#2}
        \end{minipage}
    }\\
    }
    \end{center}
}
\newcommand{\inred}[1]{{\color{BrickRed}\sf\textbf{\textsc{#1}}}}
\ifnotes
\newcommand{\note}[1]{\frameit{Note}{#1}}
\newcommand{\todo}[1]{\frameit{To-do}{#1}}
\newcommand{\inote}[1]{\inred{$<<<${#1}$>>>$}}
\else
\newcommand{\note}[1]{}
\newcommand{\todo}[1]{}
\newcommand{\inote}[1]{}
\fi

\newcommand{\NN}{{\mathbb N}}

\newtheorem{theorem}{Theorem}[section]
\newtheorem{proposition}[theorem]{Proposition}

\newtheorem{definition}[theorem]{Definition}

\newcommand{\be}{\begin{equation}}
\newcommand{\ee}{\end{equation}}

\newcommand{\comment}[1]{}

\usepackage{setspace}

\begin{document}



\title{Finding All Allowed Edges in a Bipartite Graph}

\author{
\begin{tabular}{c}
Tamir Tassa\thanks{Department of Mathematics and Computer Science, The Open University, 1 University Road, Ra'anana, Israel 43537. Telephone:
+972-9-7781272. Fax: +972-9-7780706. Email: {\tt tamirta@openu.ac.il}}\\
\end{tabular}
}

\maketitle

\begin{abstract}
We consider the problem of finding all allowed edges in a bipartite graph $G=(V,E)$, i.e., all edges
that are included in some maximum matching.
We show that given any maximum matching in the graph,
it is possible to perform this computation in linear time $O(n+m)$ (where $n=|V|$ and $m=|E|$).
Hence, the time complexity of finding all allowed edges reduces to that of finding a single maximum matching, which
is $O(n^{1/2}m)$ (Hopcroft and Karp \cite{HK}), or $O((n/\log n)^{1/2}m)$ for dense graphs with $m=\Theta(n^2)$ (Alt et al. \cite{ABMP}).
This time complexity improves upon that of the best known algorithms for the problem, which is $O(nm)$ (Costa \cite{C94} for bipartite graphs,
and Carvalho and Cheriyan \cite{CC05} for general graphs).
Other algorithms for solving that problem are randomized algorithms due to Rabin and Vazirani \cite{RV89} and Cheriyan \cite{C97}, the
runtime of which is $\tilde{O}(n^{2.376})$. Our algorithm, apart from being deterministic,
improves upon that time complexity for bipartite graphs
when $m=O(n^r)$ and $r<1.876$. In addition, our algorithm is elementary, conceptually simple, and easy to implement.
\end{abstract}

\noindent
{\bf Keywords.} Bipartite graphs, perfect matchings, maximum matchings, allowed edges

\section{Introduction}
Let $V_1$ be a set of men and $V_2$ be a set of women that are registered in a matchmaking agency.
Each of the men (women)
has an associated description and given preferences regarding the woman (man) with whom he (she) wishes to get acquainted.
Those descriptions and preferences induce a bipartite graph $G$ of compatibilities on the set of nodes $V=V_1 \cup V_2$.
The agency then presents to the clients the links that are relevant to them, so that
they can choose whom to meet.
Some of those links may not be extended to a maximum matching in the bipartite graph. The agency has an incentive to identify those links upfront
and not offer them to its clients, since if they would be offered and one of them would be successful, it would prevent the possibility of
achieving a maximum matching.
Therefore, given the full bipartite graph $G$, it is needed to remove from it all edges that are not part of a maximum matching,
or, alternatively speaking, it is needed to compute the subgraph of $G$ that equals the
union of all maximum matchings in $G$.

In order to formalize this problem, we recall some of the basic terminology regarding matchings in graphs.

\newpage
\begin{definition}\label{def1} Let $G=(V,E)$ be a graph.
\begin{itemize}
\item A matching in $G$ is a collection $M \subseteq E$ of non-adjacent edges.
\item A matching $M$ is called maximal if it is not a proper subset of any other matching.
\item A maximal matching $M$ is called a maximum (cardinality) matching if there does not exist a matching with a greater cardinality.
\item A maximum matching $M$ is called perfect if it covers all of the nodes in $V$.
\item An edge $e \in E$ is an allowed edge if it is included in some maximum matching.
\end{itemize}
\end{definition}

We study here one of the fundamental problems in matching theory, namely,
the identification of all allowed edges in a given graph $G$.
We focus here on bipartite graphs, i.e. $G=(V,E)$ where $V=V_1 \cup V_2$, $V_1 \cap V_2 = \emptyset$, and $E \subseteq V_1 \times V_2$,
and devise an efficient algorithm for that problem in such graphs.

\smallskip
Our main contributions herein are summarized as follows:
In case we are given one of the maximum matchings in the bipartite graph $G=(V,E)$, Algorithm \ref{algorithm:cycles2} in
Section \ref{sec2}
finds all allowed edges in linear time $O(n+m)$, where $n=|V|$ and $m=|E|$. (In the application that motivated this study there was indeed
one ``natural" maximum matching that was known without any computation, see Section \ref{epilogue}.)
Without such prior knowledge on the graph, we may first find any maximum matching in the graph, and then proceed to apply
Algorithm \ref{algorithm:cycles2} in order to complete the computation in additional linear time.
Algorithm \ref{algorithm:1p} implements this approach. Its runtime is dominated by the runtime of the procedure for finding
a maximum matching.

\begin{algorithm}[h!!]
\caption{\label{algorithm:1p} Finding all allowed edges in a bipartite graph $G=(V,E)$.}
\begin{algorithmic}[1]
\INPUT A bipartite graph $G=(V,E)$.
\OUTPUT All allowed edges in $E$.
\STATE Find a maximum matching in $G$, say $M$.
\STATE Using $M$, invoke Algorithm \ref{algorithm:cycles2} in order to find all allowed edges in $G$.
\end{algorithmic}
\end{algorithm}

The Hopcroft-Karp algorithm \cite{HK} offers the best known worst-case performance for finding a maximum matching
in a bipartite graph, with a runtime of $O(n^{1/2}m)$.
For dense bipartite graphs, a slightly better alternative exists:
An algorithm by Alt et al. \cite{ABMP} finds a maximum matching in a bipartite graph in time
$O\left(n^{3/2}\left(\frac{m}{\log n} \right)^{1/2}\right)$. In cases where
$m=\Theta(n^2)$, it becomes
$O((n/\log n)^{1/2}m)$, whence it is a $({\log n})^{1/2}$-factor faster than the Hopcroft-Karp algorithm.

A faster algorithm for finding a maximum matching in bipartite graphs was recently proposed by Goel et al. \cite{GKK10} for the special case
of $d$-regular graphs (graphs in which all nodes have degree $d$). However, in such graphs all edges are allowed (see more on that in
Section \ref{disc2}).

The paper is organized as follows. The algorithm for finding all allowed edges from any maximum matching in the bipartite graph is given in Section \ref{sec2}.
In Section \ref{disc3} we consider the problem in a dynamic setting and discuss an efficient implementation in such settings.
In Section \ref{related} we discuss related work and compare our algorithm to the leading algorithms. We conclude in Section \ref{epilogue}
in which we describe problems from data privacy that motivated this study and suggest future research.
In the Appendix we comment on the size of matching to which a given set of non-adjacent edges can be extended
(Section \ref{disc1}), and prove
that in regular bipartite graphs all edges are allowed (Section \ref{disc2}).

\section{A linear time algorithm for finding all allowed edges from any maximum matching}\label{sec2}
We begin by considering in Section \ref{sec21}
the case of balanced graphs ($|V_1|=|V_2|$) with a perfect matching.
In Section \ref{sec22} we consider the general case.

\subsection{Bipartite graphs with a perfect matching}\label{sec21}
Let us denote the nodes of the graph as follows: $V_1 = \{ v_1,\ldots,v_{\tilde{n}}\}$ and $V_2 = \{v'_1,\ldots,v'_{\tilde{n}}\}$ (here, $\tilde{n}=\frac{n}{2}$).
The graph $G$ is assumed to have at least one known perfect
matching; without loss of generality, we assume that it
is $M:=\{(v_1,v'_1),\ldots,(v_{\tilde{n}},v'_{\tilde{n}}) \}$.

\begin{definition}\label{defcyc} A set of $\ell \geq 1$ edges in the graph $G$ is called
an {\em alternating cycle} (with respect to $M$)
if there exist $\ell$ distinct indices, $1 \leq i_1,\ldots,i_\ell \leq \tilde{n}$,
such that the $\ell$ edges are
$(v_{i_1},v'_{i_2})$, $(v_{i_2},v'_{i_3})$,$\ldots$, $(v_{i_{\ell-1}},v'_{i_\ell})$,$(v_{i_{\ell}},v'_{i_1}) $.
\end{definition}

\begin{theorem}\label{prp34}
Let $G=(V,E)$ be a bipartite graph where $V=V_1 \cup V_2$,
$V_1 = \{ v_1,\ldots,v_{\tilde{n}}\}$, $V_2 = \{v'_1,\ldots,v'_{\tilde{n}}\}$, and $E \subseteq V_1 \times V_2$. Assume further
that $M:=\{(v_1,v'_1),\ldots,(v_{\tilde{n}},v'_{\tilde{n}}) \} \subseteq E$.
Then an edge $e \in E$ is allowed if and only if
it is included in an alternating cycle.
\end{theorem}

\begin{proof} Assume that $e=(v_{i_1},v'_{i_2})$ is part of an alternating cycle, say
$ (v_{i_1},v'_{i_2})$, $(v_{i_2},v'_{i_3})$,$\ldots$, $(v_{i_{\ell-1}},v'_{i_\ell})$, $(v_{i_{\ell}},v'_{i_1}) $.
If we augment this alternating cycle with the ${\tilde{n}}-\ell$ edges $(v_i,v'_i)$ for all $ i \notin \{ i_1,\ldots,i_\ell\}$
we get a perfect matching. Hence, $e$ is allowed.

Assume next that $e=(v_{i_1},v'_{i_2})$ is allowed. Then it is included in some perfect matching $M'$.
We proceed to define a sequence of edges $S_e \subseteq  M'$ in the following manner: $e_1:=e=(v_{i_1},v'_{i_2})$; then, for all $j \geq 1$, if $e_j=(v_{i_j},v'_{i_{j+1}})$,
we set $e_{j+1}:=(v_{i_{j+1}},v'_{i_{j+2}})$, where $v'_{i_{j+2}}$ is the node that $M'$ matches with $v_{i_{j+1}}$.
Since $M'$ is finite, the sequence must repeat itself at some point. Namely, there must exist a minimal index $j_0 \geq 2$
such that $e_{j_0}$ equals one of the previous edges in $S_e$. It is easy to see that $e_{j_0}$ must coincide with $e_1$ since if $e_{j_0}$ would coincide
with, say $e_2$, then $M'$ would have included two different edges that are incident to $v_{i_2}$. But then the sequence $S_e$ up to that point,
$S_e=\{e_1=e,e_2,\ldots,e_{j_0-1}\}$ is an alternating cycle with respect to $M$. Hence, every allowed edge must be contained in an alternating cycle.
\end{proof}

Next, we define the directed
graph $H=(U,F)$ that is induced by the bipartite graph $G=(V,E)$. In the directed graph $H=(U,F)$ the set of nodes
is $U=\{u_1,\ldots,u_{\tilde{n}}\}$ and $(u_i,u_j) $ is a directed edge in that graph
if and only if $i \neq j$ and $(v_i,v'_j) \in E$.
See example of a bipartite graph $G$ and the corresponding directed graph $H$ in Figure \ref{bipart1}.

\begin{figure}[h!!!]
\begin{center}
{\includegraphics[scale=0.6]{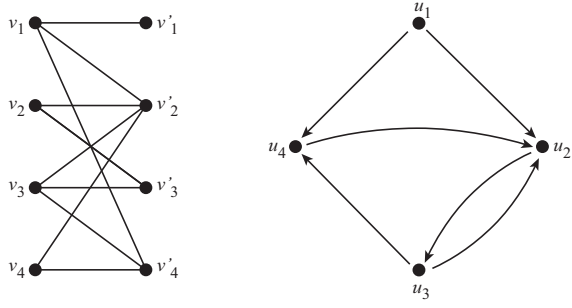}}
\caption{\label{bipart1} A bipartite graph $G$ and the corresponding directed graph $H$}
\end{center}
\end{figure}

It is easy to see that, in view of Theorem \ref{prp34},
an edge $(v_i,v'_j) \in E$ is allowed in $G$ if and only if $i=j$ or the edge $(u_i,u_j) \in F$ is part of a directed cycle in $H$.
For example, the edge $(v_3,v'_4)$ is allowed in $G$ that is given in Figure \ref{bipart1}
since it is part of the perfect matching $\{(v_1,v'_1),(v_2,v'_3),(v_3,v'_4),(v_4,v'_2)\}$ (in which the
last three edges are an alternating cycle), and indeed the corresponding edge $(u_3,u_4)$ is part of a directed cycle of length 3 in $H$.

Therefore, the problem of finding all allowed edges in $G$ reduces to the problem of finding all edges in the directed graph $H$
that are part of a directed cycle. This may be achieved as follows: First, one has to find all strongly connected components of $H$;
namely, all maximal strongly connected subgraphs of $H$ (recall that a directed graph is strongly connected if
there is a path from each node in the graph to every other node).
If each strongly connected component is contracted to a single node, the resulting graph is a directed acyclic graph.
Consequently, a given edge in $H$ is part of a directed cycle if and only if it connects two nodes in the same strongly
connected component.
There are several efficient algorithms for finding the strongly connected components of a given directed graph.
Tarjan's algorithm \cite{Tar} and Cheriyan-Mehlhorn-Gabow algorithm \cite{CM} are both equally efficient with a linear
runtime. Finally, we use those findings to mark all allowed edges in the original bipartite graph $G$.
Algorithm \ref{algorithm:cycles} does all of the above. Its runtime is $O(n+m)$.

\begin{algorithm}[h!!]
\caption{\label{algorithm:cycles} Finding all allowed edges in a bipartite graph, given a perfect matching.}
\begin{algorithmic}[1]
\INPUT A bipartite graph $G=(V,E)$ where $V=V_1 \cup V_2$,
$V_1=\{v_1,\ldots,v_{\tilde{n}}\}$, $V_2=\{v'_1,\ldots,v'_{\tilde{n}}\}$, $E \subseteq V_1 \times V_2$,
and for all $1 \leq i \leq {\tilde{n}}$, $(v_i,v'_i) \in E$.
\OUTPUT All allowed edges in $E$.
\STATE Mark all edges in $E$ as not allowed.
\STATE Construct the directed graph $H=(U,F)$ that corresponds to $G$.
\STATE Find all strongly connected components of $H$.
\FOR {all edges $(u_i,u_j) \in F$}
\IF {$u_i$ and $u_j$ belong to the same strongly connected component in $H$}
\STATE Mark the edge $(v_i,v'_j) \in E$ as allowed.
\ENDIF
\ENDFOR
\STATE Mark all edges $(v_i,v'_i) \in E$ as allowed, $1 \leq i \leq {\tilde{n}}$.
\end{algorithmic}
\end{algorithm}

\subsection{General bipartite graphs}\label{sec22}
After dealing with balanced bipartite graphs that have a perfect matching, we turn our attention to the general case.
Hereinafter we let $G=(V,E)$ be a bipartite graph where the two parts of the
graph are $V_1 = \{ v_1,\ldots,v_{n_1}\}$ and $V_2 = \{v'_1,\ldots,v'_{n_2}\}$, $n_1 \leq n_2$, and the maximum matchings are of size $t \leq n_1$.
We may assume that $t<n_2$ since if $t=n_2$, $G$ is a balanced bipartite graph with a perfect matching, and we already solved the problem for such graphs.
Let $M$ be a given maximum matching in $G$. We may assume, without loss of generality, that
$M:=\{(v_1,v'_1),\ldots,(v_{t},v'_{t}) \} $.

\begin{definition} A node $v_i$ or $v'_i$ is called an $M$-upper node if $i \leq t$ and an $M$-lower node if $i>t$.
An edge $(v_i,v'_j)$ in $G$ is called an $M$-upper edge if it connects two $M$-upper nodes; all other edges are called $M$-lower edges.
\end{definition}

In other words, a node is called $M$-upper if it is covered by $M$, and called $M$-lower otherwise.
(The
terminology simply reflects the fact that the so-called $M$-upper
nodes appear in the upper part of the graph.) Consider, for
example, the graph in the left of Figure \ref{bipart2}. In that
graph $n_1=n_2=4$, $t=3$, and $M=\{(v_1,v'_1),(v_2,v'_2),(v_3,v'_3)\}$; the graph has six $M$-upper nodes and two
$M$-lower ones. The edges $(v_3,v'_4)$ and $(v_4,v'_1)$ are $M$-lower
edges, while all other five edges are $M$-upper.

\begin{figure}[h!!!]
\begin{center}
{\includegraphics[scale=0.6]{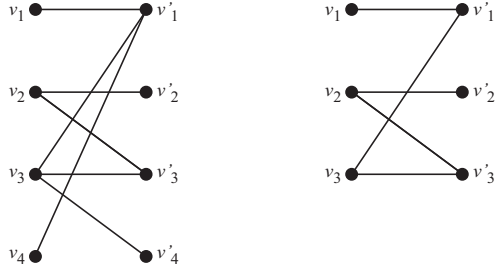}}
\caption{\label{bipart2} A bipartite graph $G$ and the corresponding restriction $G_u$}
\end{center}
\end{figure}

For simplicity, we shall use hereinafter the terms upper and lower without specifying the prefix $M$. It should always be understood that a node or an edge
are upper or lower with respect to the given $M$.

Our first observation towards a classification of all allowed edges in such a graph is as follows.
\begin{proposition}\label{prop45} Let $G=(V,E)$ be a bipartite graph as described above. Then:
\begin{enumerate}
\item The graph has no lower edges that connect two lower nodes.
\item All lower edges are allowed.
\end{enumerate}
\end{proposition}

\begin{proof} The first claim is obvious since if there was a lower edge $(v_i,v'_j)$ where $i,j>t$, then $M \cup \{(v_i,v'_j)\}$ would have been a matching, thus
contradicting the assumed maximality of $M$. As for the second claim, let $(v_i,v'_j)$ be a lower edge with $i>t$. In view of the first claim, $v'_j$ must be an upper
node.
Therefore, $\left( M \cup \{(v_i,v'_j)\} \right) \setminus \{ (v_j,v'_j)\}$ is also a maximum matching, whence $(v_i,v'_j)$ is an allowed edge.
\end{proof}

As Proposition \ref{prop45} determines that all lower edges are allowed, it remains to identify the allowed edges among the upper edges.
Let $G_u$ denote the restriction of $G$ to the upper nodes, $\{v_1,v'_1,\ldots,v_{t},v'_{t} \}$.
We distinguish between two types of allowed edges among the upper edges:

\begin{definition} An upper edge is an allowed edge of type I (with respect to $M$) if it is included in a maximum matching that consists only of upper edges.
An upper edge is an allowed edge of type II (with respect to $M$) if it is allowed, but all maximum matchings that include it include also a lower edge.
\end{definition}

For example, consider the graph $G$ on the left of Figure \ref{bipart2} and the corresponding $G_u$ on the right. The edge $(v_3,v'_1)$
is not allowed, but all other six edges are allowed:
\begin{enumerate}
\item The edges $(v_3,v'_4)$ and $(v_4,v'_1)$ are allowed since they are lower edges. (For example, $(v_3,v'_4)$ may be extended to a maximum matching with
$(v_1,v'_1)$ and $(v_2,v'_2)$.)
\item The edges $(v_1,v'_1)$, $(v_2,v'_2)$, and $(v_3,v'_3)$ are allowed edges of type I since they are allowed edges also in $G_u$.
\item The edge $(v_2,v'_3)$ is an allowed edge of type II since it is allowed (together with $(v_1,v'_1)$ and $(v_3,v'_4)$ it forms a maximum matching), but it is not
allowed in $G_u$.
\end{enumerate}

All allowed edges of type I can be identified by applying Algorithm \ref{algorithm:cycles} on the subgraph $G_u$ (since $G_u$ is a balanced graph with a perfect matching).
It remains only to identify the allowed edges of type II.

\begin{definition}\label{defpath} A set $P$ of $\ell-1 \geq 1$ upper edges in $G$ is called
an {\em upper alternating path} (with respect to $M$)
if there exist $\ell$ distinct indices, $1 \leq i_1,\ldots,i_\ell \leq t$,
such that
\be P=\{ (v_{i_1},v'_{i_2}), (v_{i_2},v'_{i_3}), \ldots, (v_{i_{\ell-1}},v'_{i_\ell}) \} \,.\label{bip1}\ee
If $i_0>t$ and $\lambda=(v_{i_0},v'_{i_1})$ is a lower edge in $E$ then $\{\lambda\} \cup P$ is called a {\em left-augmented alternating path}.
If $i_{\ell+1}>t$ and $\rho=(v_{i_\ell},v'_{i_{\ell+1}})$ is a lower edge in $E$ then $ P \cup \{\rho\}  $ is called a {\em right-augmented alternating path}.
Finally, we let ${\cal AP}$ denote the set consisting of all alternating paths of all three sorts --- upper, left- and right-augmented.
\end{definition}

For example, $P=\{(v_2,v'_3)\}$ is an upper alternating path of length 1 in the graph $G$ of Figure \ref{bipart2}, while $P \cup \{(v_3,v'_4)\}$ is a corresponding right-augmentation.

\begin{theorem}\label{thm1} Any allowed edge of type II with respect to $M$ must be contained in
an either left- or right-augmented alternating path in $G$ with respect to $M$. On the other hand, an edge that is contained in an either
left- or right-augmented alternating path is allowed.
\end{theorem}

\begin{proof} Let $e=(v_i,v'_{i'})$ be an upper edge which is an allowed edge of type II with respect to $M$.
Let $M_e$ be a maximum matching that includes $e$. Define
$${\cal AP}_e = \{ P \in {\cal AP}: P \subseteq M_e \mbox{ and } e \in P \}\,;$$
namely, ${\cal AP}_e$ consists of all alternating paths (upper, left- and right-augmented ones) that are contained in $M_e$ and include $e$.
Clearly, ${\cal AP}_e$ is nonempty since it includes the upper alternating path of length 1 that consists only of $e$.
Hence, we may select a path $P \in {\cal AP}_e$ of maximal length
among all paths in ${\cal AP}_e$.
We claim, and prove below, that $P$ cannot be an upper alternating path. Hence, $P$ must be either left- or right-augmented alternating path.
Since $P$ includes $e$, that will settle the necessary condition in the theorem: Any allowed edge of type II with respect to $M$ must be contained in
an either left- or right-augmented alternating path with respect to $M$.

Assume that $P$ is an upper alternating path. Denoting its length by $\ell-1$, $P$ must take the form
$$ P=\{ (v_{i_1},v'_{i_2}), (v_{i_2},v'_{i_3}), \ldots, (v_{i_{\ell-1}},v'_{i_\ell}) \} \,,$$
where $1 \leq i_1,\ldots,i_\ell \leq t$.
We claim that the matching $M_e$ does not match neither $v'_{i_1}$ nor $v_{i_\ell}$. Assume, towards contradiction, that $M_e$ includes an edge $(v_j,v'_{i_1})$.
Then $j$ cannot be any of the indices outside the set $\{i_1,\ldots,i_\ell\}$, as that would contradict the assumed maximality of $P$.
In addition, $j$ cannot be any of the indices in $\{i_1,\ldots,i_{\ell-1}\}$ since then $M_e$ would include two edges that are adjacent to $v_j$.
Finally, $j$ cannot be $i_\ell$ since then $e$ would have been contained in an alternating
cycle in $G_u$ (Definition \ref{defcyc}) and then, by Theorem \ref{prp34}, $e$ would have been allowed in $G_u$, thus contradicting our assumption that $e$ is
an allowed edge of type II.
Therefore, $M_e$ does not match $v'_{i_1}$. For the same reasons, $M_e$ cannot match $v_{i_\ell}$. But then, consider the set of edges
$$ M'_e = \left( M_e \setminus P\right) \cup \{ (v_{i_k},v'_{i_k}) : 1 \leq k \leq \ell \} \,.$$
In view of the above discussion, $M'_e$ is a matching (since $M_e \setminus P$ does not cover any of the nodes $v_{i_k}$ and $v'_{i_k}$, $1 \leq k
\leq \ell$). But as $|M'_e|=|M_e|+1$, it contradicts our assumption that $M_e$ is a maximum matching.

Next, we turn to prove the sufficiency of the condition; namely, that an edge that is contained in an
either left- or right-augmented alternating path with respect to $M$ must be allowed in $G$. Let
$$ P=\{ (v_{i_0},v'_{i_1}), (v_{i_1},v'_{i_2}), (v_{i_2},v'_{i_3}), \ldots, (v_{i_{\ell-1}},v'_{i_\ell}) \} $$
be a left-augmented path in $G$ (namely, $i_0>t$ and $i_k \leq t$ for all $1 \leq k \leq \ell$). Then
\be M_0 = P \cup \{ (v_j,v'_j) : j \in \{1,\ldots,t\} \setminus \{i_1,\ldots,i_\ell\}\} \label{m0}\ee
is a matching of cardinality $t$. Hence, all edges in $P$ are included in a maximum matching, whence they are all allowed. (One of those allowed edges, namely
$(v_{i_0},v'_{i_1})$, is a lower edge, while all the other upper edges may be allowed edges of either type I or II). The proof for right-augmented paths is similar.
\end{proof}

Hence, in order to identify all allowed edges of type II with respect to $M$, we have to scan all left- and right-augmented alternating paths and mark all upper
edges along them as being allowed
(some of those allowed edges may be of type I, whence they will be "discovered" twice).
To do that, we use again the representation of $G$ as a directed graph $H=(U,F)$.

\begin{definition} Let $G=(V,E)$ be a bipartite graph where the two parts of the
graph are $V_1 = \{ v_1,\ldots,v_{n_1}\}$ and $V_2 = \{v'_1,\ldots,v'_{n_2}\}$, and $n_1 \leq n_2$.
The corresponding left-to-right directed graph is $H_{LR}=(U,F_{LR})$
where $U=\{u_1,\ldots,u_{n_2}\}$ and $(u_i,u_j) \in F_{LR}$ if and only if $(v_i,v'_j) \in E$. The
right-to-left directed graph is $H_{RL}=(U,F_{RL})$
where $U=\{u_1,\ldots,u_{n_2}\}$ and $(u_i,u_j) \in F_{RL}$ if and only if $(v_j,v'_i) \in E$.
\end{definition}

With those definitions, an edge $(v_i,v'_j)$ is on a left-augmented (resp. right-augmented)
alternating path if and only if the edge $(u_i,u_j) $ is reachable in $H_{LR}$ (resp. $H_{RL}$) from one
of the nodes $u_k$, where $t+1 \leq k \leq n_2$. Hence, all that is needed is to perform a BFS of $H_{LR}$ (resp. $H_{RL}$) from those nodes
and for every edge $(u_i,u_j) $ that we visit along this scan, mark the corresponding edge $(v_i,v'_j)$ (resp. $(v_j,v'_i)$) in $G$ as an allowed edge.
The runtime of such a scan is linear. Algorithm \ref{algorithm:cycles2} summarizes the procedure of finding all allowed edges in a general bipartite graph.

\begin{algorithm}[h!!]
\caption{\label{algorithm:cycles2} Finding all allowed edges in a general bipartite graph, given a maximum matching.}
\begin{algorithmic}[1]
\INPUT A bipartite graph $G=(V,E)$ where $V=V_1 \cup V_2$,
$V_1=\{v_1,\ldots,v_{n_1}\}$, $V_2=\{v'_1,\ldots,v'_{n_2}\}$, $E \subseteq V_1 \times V_2$,
and $\{(v_i,v'_i): 1 \leq i \leq t\}$ is a given maximum matching in $G$.
\OUTPUT All allowed edges in $E$.
\STATE Mark all lower edges as allowed.
\STATE Mark all upper edges as not allowed.
\STATE Apply Algorithm \ref{algorithm:cycles} on the restriction $G_u$ of $G$ to $\{v_1,v'_1,\ldots,v_{t},v'_{t} \}$, thus marking all upper edges that are allowed in $G_u$.
\STATE Construct the left-to-right directed graph $H_{LR}$.
\STATE Add to $H_{LR}$ a new source node $u_0$ and connect it to each of the nodes
$u_k$, where $t+1 \leq k \leq n_1$.
\STATE Apply a BFS on the graph $H_{LR}$ starting from $u_0$. For each edge $(u_i,u_j)$ that the BFS visits,
mark the corresponding edge $(v_i,v'_j)$ in $G$ as allowed.
\STATE Construct the right-to-left directed graph $H_{RL}$.
\STATE Add to $H_{RL}$ a new source node $u_0$ and connect it to each of the nodes
$u_k$, where $t+1 \leq k \leq n_2$.
\STATE Apply a BFS on the graph $H_{RL}$ starting from $u_0$. For each edge $(u_i,u_j)$ that the BFS visits,
mark the corresponding edge $(v_j,v'_i)$ in $G$ as allowed.
\end{algorithmic}
\end{algorithm}

\section{The dynamic setting}\label{disc3}
Going back to the opening example of the matchmaking agency, if one of the allowed edges materializes and
another couple of clients leaves happily the matching game, the agency has to update the set of allowed edges in the reduced graph.
Here we comment on how to cope with such dynamic updates efficiently. Before doing so, we describe another interesting example.

A domino board is a bounded region $D$ in the Eucledean plane which is the union of unit squares ($S_{i,j}:=[i,i+1] \times [j,j+1]$ for some $i,j \in \NN$) that are connected
in the sense that there exists a path between the centers of every two unit squares that is fully contained within the interior of $D$.
Each unit square $S_{i,j}$ is colored white in case $i+j$ is even, or black otherwise.
A perfect tiling of the region $D$ is a cover of $D$
by $|D|/2$ non-overlapping dominos ($|D|$ being the area of $D$), where dominos are shapes formed by the union of two unit squares meeting edge-to-edge.

Now, let us consider the following computerized game: A domino region that has a perfect tiling is presented to the player, who needs to tile it by placing
dominos on it, one at a time.
If he places
a domino in a location that prevents any completion of the tiling, the computer issues an alert and then the player has a chance to try again.
The player wins if he was able to complete the tiling successfully with a number of
bad moves below some predefined threshold.

Here too, there is an underlying bipartite graph, where each node represents a white or a black square in $D$. Each square is connected to any of its neighboring
squares of the opposite color. A perfect tiling of $D$ is a perfect matching in that graph.
The bad moves correspond to edges in that graph that are not allowed.
In each step in the game, the player places one domino; it is then
needed to remove the corresponding nodes and adjacent edges from the underlying graph and update the list of allowed edges in the reduced graph.

\smallskip
Following the two motivating examples, we now formulate the problem in the dynamic setting.
Let $G$ be a bipartite graph as discussed in Section \ref{sec22}. Assume that we already found
a maximum matching in $G$, say $M:=\{(v_1,v'_1),\ldots,(v_{t},v'_{t}) \} $, and identified all allowed edges in $G$.
In addition, we may assume
that each allowed edge
is marked by its type: In balanced graphs with a perfect matching all allowed edges are of the same type, but in other graphs we distinguish between allowed edges
that are lower edges, and allowed edges that are upper edges of type I or of type II.
Let $(v_i,v'_j)$ be an allowed edge in $G$ and let $G'$ be the graph that is obtained from $G$ by removing $v_i$, $v'_j$, and all adjacent edges.
We wish to find a maximum matching in $G'$ in an efficient manner, in order to proceed and identify in linear time all allowed edges in $G'$.
We separate the discussion to three cases:
\begin{itemize}
\item If $(v_i,v'_j)$ is an allowed edge of type I, we look in
the directed graph $H_{LR}$ for a path from $u_j$ to $u_i$. Assume that the path is
$$(u_j,u_{j_1}),(u_{j_1},u_{j_2}),\ldots,(u_{j_{\ell-1}},u_{j_\ell}), (u_{j_\ell},u_i)\,.$$
Then it is easy to see that
$$ M \setminus \{ (v_k,v'_k): k \in \{i,j,j_1,\ldots,j_\ell\}\} \cup \{(v_j,v'_{j_1}), (v_{j_1},v'_{j_2}),\ldots,(v_{j_{\ell-1}},v'_{j_\ell}), (v_{j_\ell},v'_i)\} $$
is a maximum matching in $G'$.
\item If $(v_i,v'_j)$ is an allowed edge of type II, then Theorem \ref{thm1} implies that
either the edge $(u_i,u_j)$ is contained in a path in $H_{LR}$ that ends with a lower node,
or the edge $(u_j,u_i)$ is contained in a path in $H_{RL}$ that ends with a lower node.
Hence, we may apply a BFS on $H_{LR}$ starting in $u_i$ in order to look for a path that connects it to a lower node;
if none is found, a similar BFS on $H_{RL}$, starting in $u_j$, is guaranteed to find a path that reaches a lower node.
Once such a path is found, we may reconstruct the corresponding maximum matching $M_0$ in $G$ that contains the edge $(v_i,v'_j)$ (see Eq. (\ref{m0})).
Clearly, $M_0 \setminus \{(v_i,v'_j)\}$
is a maximum matching in $G'$.
\item Finally, if $(v_i,v'_j)$ is a lower edge then $M \setminus \{(v_i,v'_i)\}$ is a maximum matching in $G'$ in case $i \leq t$, or
$M \setminus \{(v_j,v'_j)\}$ is, in case $j \leq t$.
\end{itemize}

\section{Related work}\label{related}
Costa \cite{C94} presented an algorithm for decomposing the edge set of a bipartite graph $G=(V,E)$ into three disjoint partitions, $E=E_1 \cup E_w \cup E_0$, where
$E_1$ contains all allowed edges that belong to all maximum matchings in $G$, $E_w$ contains all other allowed edges, and $E_0$ consists of the non-allowed edges.
The runtime of
her algorithm is $O(nm)$.
Carvalho and Cheriyan \cite{CC05} designed an algorithm with the same time complexity for finding all allowed edges in
general graphs, using well-known results on efficient
implementations of Edmonds' maximum-matching algorithm and other results from the matching folklore.

Rabin and Vazirani \cite{RV89} designed a simple randomized algorithm for finding all allowed edges in perfectly-matchable general graphs.
We outline their algorithm: Let $G=(V,E)$ be the input graph. Its Tutte matrix is defined as the following $n \times n$ matrix,
$$ A_{i,j}=\left\{ \begin{array}{ll} x_{i,j} & (v_i,v_j)\in E ~,~~ i>j \\ -x_{i,j} & (v_i,v_j)\in E ~,~~ i<j \\ 0 & \mbox{otherwise} \end{array} \right. \,, $$
where $x_{i,j}$ are indeterminates. Let $p>n^4$ be any prime, $S(A)$ be a random substitition of all $x_{i,j}$ with elements from $GF(p)$, and $B = S(A)^{-1}$.
Then with probability at least $e^{-2/n^2}$, the matrix $B$ identifies the set of allowed edges, in the sense that
$(v_i,v_j)$ is an allowed edge if and only if $B_{i,j} \neq 0$.
The runtime of this algorithm is determined by matrix inversion time and it equals $\tilde{O}(n^{2.376})$. Cheriyan \cite{C97} designed a similar algorithm
with the same runtime that applies for any general graphs (i.e., not necessarily ones with a perfect matching).

Our algorithm is deterministic. Its runtime is determined by the time to find a single maximum matching,
which is $O(n^{1/2}m)$ (Hopcroft and Karp \cite{HK}) or $O((n/\log n)^{1/2}m)$ for dense graphs with $m=\Theta(n^2)$ (Alt et al. \cite{ABMP}).
Hence, it improves upon that of \cite{C94,CC05} by a factor of $n^{1/2}$ (or $(n \log n)^{1/2}$ for dense graphs, where $m=\Theta(n^2)$).
In comparison to the randomized algorithms, our algorithm is faster when $m = O(n^r)$ and $r < 1.876$, since then its runtime is $O(n^{r+1/2})$, where $r+1/2 < 2.376$.
In addition, our algorithm is elementary, conceptually simple, and easy to implement.
Another advantage of our algorithm in comparison to the above mentioned algorithms is that in cases where there is one "natural" maximum matching, which
is known without invoking a costly procedure for finding one, the runtime of our algorithm is only $O(n+m)$. The runtime of all of the above mentioned algorithms
remains the same even in such cases. (Examples of cases where one maximum matching is known upfront are described in Section \ref{epilogue}.)

Identifying all allowed edges is equivalent to computing the union of all maximum matchings in the graph.
Another line of research concentrated on algorithms for enumerating all maximum matchings in bipartite graphs, e.g. \cite{FM94,Uno97,Uno01}.
Naturally, the runtime of such algorithms depends on the number of the counted matchings (perfect matchings in \cite{FM94,Uno01} and maximal, maximum, or perfect
matchings in \cite{Uno97}). We recall in this context that counting the number of perfect matchings in a bipartite
graph is equivalent to computing the permanent of
a $\{0, 1\}$-matrix, a problem that is known to be in \#P-complete \cite{val}.

\section{Epilogue}\label{epilogue}
This study was motivated by a problem of providing anonymity in databases \cite{GMT}. Consider a data owner that has a database
of records that hold information about individuals in some population. He needs to publish the database for the purpose of analysis and public scrutiny.
In order to thwart adversarial attempts to identify the individuals behind some of those records, he generalizes the data in those records so that each generalized
record would be consistent with the personal information of more than one individual.
In order to ensure a minimal level of anonymity, the data owner wishes to guarantee that each
generalized record would be consistent with the personal information of no less than $k$ individuals, for some given threshold $k$.
As the adversary knows that there is a one-to-one mapping
between the population of individuals and the records in the generalized database, this setting may be modeled by a bipartite graph that has a perfect matching.
Hence, the adversary knows that all edges
in the bipartite graph that are not allowed can be ignored, as they cannot stand for a true linkage between an individual and his
record in the table. Therefore, the data owner needs to make sure that even after
the removal of such non-allowed edges from the bipartite graph, the anonymity constraint is still respected. Since the data owner knows the true perfect matching (because
he has the raw data and he performs on it the generalization), he
may use Algorithm \ref{algorithm:cycles} to identify all non-allowed edges efficiently and then verify that the generalized version of the database that he created
provides the required level of anonymity.

Another example from the realm of privacy and anonymity is the work of Yao et al. \cite{YaoWJ05}. In that work, they consider the case
in which multiple views of the same database need to be released and it is needed to check whether their release would violate
some privacy constraint ($k$-anonymity). The algorithm which they suggest needs to identify all edges in a bipartite graph that are part of
some complete matching. (A matching in a bipartite graph is called complete if it covers all nodes in the possibly smaller part of the graph.)
They suggest to do it using a na\"{\i}ve approach: Testing each edge in the graph whether its removal leaves a graph that still has a complete matching.
Such an approach entails invoking the Hopcroft-Karp algorithm for each edge, whence its time complexity is $O(n^{1/2}m^2)$. However, in the case that was
studied there, the data owner knows one complete matching in the underlying bipartite graph. Whence, he may use Algorithm
\ref{algorithm:cycles2} in order to perform the same task in $O(n+m)$. The improvement factor in runtime is
$O(n^{1/2}m)$. As in such applications the graphs tend to be very large, such an improvement offers a practical solution instead
of an impractical one.

An interesting research problem that this study suggests is to extend this approach to the general case, namely to non-bipartite graphs.
The question is whether a similar algorithm that uses one known maximum matching
can be devised in order to find all allowed edges efficiently.

\newpage
\appendix
\section{Extending a set of non-adjacent edges to a matching}\label{disc1}
The discussion in Section \ref{sec2} implies that any set of $k$ non-adjacent edges in a bipartite graph may be extended to a matching of size
at least $t-k$, where $t$ is the size of a maximum matching in $G$.

\begin{theorem}\label{lowb} Let $G=(V,E)$ be a bipartite graph where the two parts of the
graph are $V_1 = \{ v_1,\ldots,v_{n_1}\}$ and $V_2 = \{v'_1,\ldots,v'_{n_2}\}$, $n_1 \leq n_2$, and the maximum matchings are of size $t \leq n_1$.
Let $P$ be a set of $k$ non-adjacent edges in $E$. Then there exists a matching of size $t-k$ in $G$ that includes $P$.
More specifically, assume that
$P=\{ (v_{j_i},v'_{\ell_i}): 1 \leq i \leq k \}$
where $A=\{j_1,\ldots,j_k\}  \subset [n_1]$ and $A'=\{\ell_1,\ldots,\ell_k\}  \subset [n_2]$ are the corresponding subsets of $k$ distinct indices.
Then there exists a matching of size $k+t-|A \cup A'|$ in $G$ that includes $P$.
\end{theorem}

\begin{proof} As before, we may assume, without loss of generality, that one of the maximum matchings in $G$ is
$M:=\{(v_1,v'_1),\ldots,(v_{t},v'_{t}) \} $.
Let $ B = [t] \setminus (A \cup A') $ denote the set of indices of edges from $M$ that are not adjacent to any of the edges in $P$.
The size of $B$ is at least $t - |A \cup A'|$.
Clearly, $P \cup \{ (v_i,v'_i) : i \in B \}$ is a matching that includes $P$. Its size is $k + |B| \geq k+t-|A \cup A'|$.
Since $|A \cup A'| \leq |A|+|A'|=2k$, we infer that $P$ is included in a matching of size at least $t-k$.
\end{proof}

It is easy to see that the lower bound in Theorem \ref{lowb} is sharp. To exemplify that, consider the case where $M$ is a perfect matching consisting of $t$ edges, the set of $k$ non-adjacent
edges $P$ is such that $|A \cup A'|=2k$, and $E=M \cup P$. The only way to extend $P$ into a matching is by adding to it edges from $M$. The maximal number
of edges from $M$ that could be added to $P$ is precisely $t-|A \cup A'|$ in this case.

\section{All edges in a regular bipartite graph are allowed}\label{disc2}
As noted earlier, the problem of finding a maximum matching in a $d$-regular bipartite graph is easier. The maximum matchings in such graphs are perfect.
A randomized algorithm due to Goel et al. \cite{GKK10} finds them in time $O(n \log n)$. However, the problem of finding all allowed edges in such graphs is trivial,
since the set of edges in such graphs may be partitioned into a disjoint union of $d$ perfect matchings, and, consequently, all edges in such graphs are allowed.

\begin{theorem} All edges in a regular bipartite graph are allowed.
\end{theorem}

\begin{proof} In regular bipartite graphs, $G=(V,E)$, where $V=V_1 \cup V_2$, all nodes have the same degree $d$. Hence, such graphs are balanced, $|V_1|=|V_2|$.
Let $X$ be any subset of $V_1$ and let $N(X)$ denote the set of all its neighbors in $V_2$. Clearly, $|N(X)| \geq |X|$ since otherwise at least one of the nodes
in $N(X)$ would have to be of degree higher than $d$. Hence, Hall's marriage theorem \cite{Hall} implies that $G$ has a perfect matching.
Let $M_1$ be such a matching and consider the graph $G_1=(V,E \setminus M_1)$. This graph is $(d-1)$-regular, whence it too has a perfect matching, say $M_2$.
Continuing in that manner we may construct $d$ disjoint perfect matchings, $M_1,\ldots,M_d$, such that $E=\bigcup_{1 \leq i \leq d}M_i$. It follows that
all edges in $E$ are allowed.
\end{proof}

\newpage
\bibliographystyle{abbrv}
\bibliography{perfectmatch}

\end{document}